\documentclass{article}
\usepackage{amsmath}
\usepackage{graphicx}
\usepackage{algorithmicx}
\usepackage{algpseudocode}
\usepackage{amsfonts}
\usepackage{amsthm}
\usepackage{amssymb}
\usepackage{placeins}
\usepackage{tabularx}
\usepackage{mathabx}
\usepackage{xypic}
\xyoption{all}

\theoremstyle{theorem}
\newtheorem{theorem}{Theorem}[section]
\newtheorem{lemma}[theorem]{Lemma}

\newtheorem{corollary}[theorem]{Corollary}
\theoremstyle{definition}
\newtheorem{definition}[theorem]{Definition}
\newtheorem{example}[theorem]{Example}

\theoremstyle{remark}
\newtheorem{remark}[theorem]{Remark}

\title{Cyclic space-filling curves and their clustering property}
\author{Igor V. Netay
\thanks{Joint Stock "Research and production company ``Kryptonite"}
\thanks{Institute for Information Transmission Problems, Russian Academy of Sciences}
}
\begin{document}
\maketitle

    \begin{abstract}
        In this paper we introduce an algorithm of construction
        of cyclic space-filling curves.
        One particular construction provides a family
        of space-filling curves in all dimensions (H-curves).
        They are compared here with the Hilbert curve in the sense of clustering properties,
        and it turns out that the constructed curve is very close and sometimes
        a bit better than the Hilbert curve.
        At the same time, its construction is more simple and evaluation
        is significantly faster.
    \end{abstract}

    \section*{Introduction}
    \label{sec:intro}
    A space-filling curve in dimension~$d$ is a map from~$[0,1]$ into~$\mathbb{R}^d$
such that the image contains an open non-empty set or, equivalently, some
cube~$[a_1,b_1]\times\ldots\times[a_d,b_d]$ for~$a_i<b_i$, $i=1,\ldots,d$.
There are lots of space filling curve constructions,
many of them can be found in~\cite{16Curves}.

\textit{Fractal space-filling curves} are self-similar curves, i.\,e. they exhibit
similar patterns at increasingly small scales.
This similarity is called \textit{unfolding symmetry}.
If the symmetry is a composition of scale and an isometry, the curve is called \textit{affine self-similar}.

Affine self-similar curves play an important role in practice, because they can be easily
constructed and provide a way to construct locality-preserving mapping from multidimensional
data into one-dimensional space.
For instance, the applications include
\begin{itemize}
    \item geo-information systems (GIS, see~\cite{GIS}),
    \item database indices (map multidimensional data to one-dimensional disk address space, see~\cite{908985}),
    \item image compression (clustering of pixels by three-dimensional color, see~\cite{Img}),
    \item parallel processing (see~\cite{ParProc}),
    \item bandwidth reduction of digitally sampled signals (see~\cite{DSP}).
\end{itemize}

There are many sophisticated curve constructions in the literature.
The simplest curve construction which is also called Z-curve was proposed in~\cite{Z}.
Its improvement by usage of Grey coding was proposed in~\cite{Gray}.
Another method based on the Hilbert curve~\cite{Hilbert} was proposed in~\cite{SecKey}.

The first space-filling curve was discovered by G.\,Peano in~1890.
This curve is continuous in terms of Jordans's precise notion of continuity (1887).
In~1891 D.\,Hilbert discovered a general geometric construction procedure for a class of space-filling curves~\cite{Hilbert}.
It has been shown that the Hilbert curve is a continuous, surjective and nowhere differentiable mapping~\cite{ContHilb}.

We follow the procedure mentioned above with some modification making all the constructed
curves cyclic thus leading to a class of continuous space-filling curves~$S^1\to[0,1]^d$.
Constructions in~\cite{Hilbert} are based on ordering,~i.\,e. one subdivides a cube into a grid of
half-size cells being recursively divided up to $2^{nd}$ unit cells
(assuming that the initial cube has size length $2^n$),  where $n$ the \textit{depth} of subdivision.
Then an order of half-size cells is chosen and an order of unit cells is recursively defined,
thus giving  the ordering of all unit cells.
We apply a different approach including construction of oriented cycles and combining them
to a single oriented cycle on each step of the construction.
Resulting class of curves differs from the order-based space-filling curves.
Its advantage is the simpler algorithmic description of continuous curves.
In the case of order-based construction we must apply some reflections (or any of $2^d\cdot d!$ symmetries)
to sub-cells.
In the case of cycle-based construction we can avoid usage of all sub-cell symmetries
and only apply rotations and/or reversals of cycles obtaining the curve which we call \textit{H-curve}.
In terms of evaluation this can be expressed as an linear operation modulo
length of the cycle of the form $x\mapsto \pm x + c$.

Not only fractal curves are used in practice.
For example, the onion curve~\cite{Onion} and spectral curve~\cite{1260840} can give better
results than the Hilbert curve, but have a fixed space granularity.
Fractal curves have the advantage that one can select granularity and
evaluate indexes with different precision for different points.
For example, if we need only to compare the ordering of a pair of points,
we need only to evaluate indexes until the first difference occurred.

Let us describe carefully the notion of \textit{locality-preserving mappings}.
Roughly speaking, we want to construct such a mapping~$\omega$ that the closer the
images~$\omega(x),\omega(y)\in\mathbb{R}^d$ of two points are, the closer the points~$x,y\in\mathbb{R}$ are
and vice versa (like it is described in~\cite{16Curves}).
In other words, one condition is that the mapping~$\omega$ is \textit{open} and other is that $\omega$
is \textit{continuous}.
Unfortunately, there are no curves satisfying both these conditions.

If a curve is a one-to-one correspondence (we need it to be injection,
because we consider only surjective mappings
by the definition of space-filling curves), then it is a \textit{homeomorphism}
(open continuous bijection).
But the interval $[0,1]$ and a cube $[0,1]^d$ for~$d>1$ are not homeomorphic,
because $[0,1]$ without any point except $0$ and $1$ is not connected and $[0,1]^d$
without any point is connected.
Number (more precisely, cardinality) of points~$x\in X$ for a topological space~$X$
such that $X\setminus\{x\}$ is disconnected
is a topological invariant and preserves under homeomorphisms.
The subset of such points in~$[0,1]^d$ is empty for~$d>1$ and is~$(0,1)$ for~$d=1$, so we get a contradiction.

We will see below that the curves we consider are not one-to-one correspondences.
Nevertheless, there are no open continuous mappings from~$[0,1]$ to~$[0,1]^d$ for~$d>1$.

Now suppose that a curve $\omega\colon[0,1]\to[0,1]^d$ is continuous, open and is not injective.
By an equivalent definition of continuous map, it is a map such that the preimage of a closed subset is closed.
Therefore, for any point~$v\in[0,1]^d$ as a closed subset its preimage~$\omega^{-1}(v)\subseteq[0,1]$ is a closed subset of the compact~$[0,1]$,~i.\,e.
is compact and therefore includes its minimum, so~$\omega(\min(\omega^{-1}(v)))=v$.
Define~$S=\{\min(\omega^{-1}(v))\,\mid\,v\in[0,1]^d\}\subseteq [0,1]$.
From non-injectivity it follows that~$S \ne [0,1]$.
Obviously,~$0\in S$.
If~$S=[0,1)$, then we get a contradiction in the same way as before.
Otherwise $p \in (0,1) \setminus S$ exists,~i.\,e. the set~$S$ is not connected.
But the mapping~$\omega\colon S \to [0,1]^d$ is open, continuous and bijective by construction.
Therefore, $S$ and~$[0,1]^d$ are homeomorhic.
At the same time, $[0,1]^d$ is connected and $S$ is not, so we again get a contradiction.

This simple topological reasoning shows that we need to weaken the conditions on locality preserving mappings.
We can obtain continuity, but we need some other way to compare which of curves ``better preserves locality''.
Here we follow a well-known idea from~\cite{908985} (see~\S\ref{sec:clustering_property}) to compare
numbers of ``connected components of preimages of connected figures'' asymptotically for curve construction iterations.
In the same way, we perform a numerical simulation experiment and obtain results (see~\S\ref{subsec:simulation}).

Another interesting property of such maps (like Hilbert curve) is that they are \textit{measure-preserving}, i.\,e.
if~$S\subseteq[0,1]$ has one-dimensional Lebesque measure~$z$, then its image~$\{\omega(s)\,\mid\,s\in S\}$ has
$d$-dimensional Lebesque measure~$z$.
We omit the proof of this property as a simple analysis exercise.

The idea of continuity is a useful heuristic to construct curves with better locality preserving properties.
The main results of this paper are
\begin{itemize}
    \item to introduce the idea of cyclicity~(see~\S\ref{sec:constructions}),
    \item to construct an explicit cyclic curve for any dimension~(see~\S\ref{sec:hcurve}),
    \item to conduct an experiment providing an empirical evidence that the constructed curve is slightly better (or not worse) than the Hilbert curve,
    \item to show that the construction of the these curves is simpler than the construction of most widely used Hilbert curve.
    \item to show by profiling that H-curve can be evaluated essentially faster than the Hilbert curve.
\end{itemize}
So, for a number of applications this new construction method may be preferable to the Hilbert curve.

    \section{Constructions of curves}
    \label{sec:constructions}
    The generic construction process of fractal space-filling curve is usually iterative.
We need to map an interval size of~$1$ to a square size of~$1$ (or a cube of dimension~$d$).
Let us illustrate this for dimension~$d=2$.
On the first step, we divide the square into a grid of~$2\times2$ square cells, while the
interval size of~$1$ is subdivided into four equal sub-intervals
where each sub-interval matches a cell.
We say that the curve traverses the cells in the order given by the order of intervals.
Then we apply the procedure recursively to each sub-interval-cell pair, so that within each cell,
the curve makes a similar traversal up to symmetries of the whole cell.
The symmetries are needed to make each cell's first sub-cell touching the previous cell's last subcell.
This condition after the going to the limit gives us continuity.
Let us present a more detailed algorithm.

Usually the construction of fractal space-filling curves consists of the following steps:
\begin{itemize}
    \item divide a cube of dimension~$d$ into the grid of~$2^d$ half-size cells and
    match them to~$2^d$ equal sub-intervals of interval;
    \item perform the iteration steps: given a matching between $2^{dn}$ cells in the
    cube and~$2^{nd}$ sub-intervals in the interval,
    \begin{itemize}
        \item subdivide each cell into the grid of~$2^{nd}$ sub-cells
        (and maybe apply some symmetry of cell) and match them to sub-intervals
        of the corresponding sub-interval,
        \item join this to the matching between grid of~$2^{d(n+1)}$ cells of the cube
        and~$2^{d(n+1)}$ sub-intervals of the interval.
    \end{itemize}
\end{itemize}

Here we introduce another algorithm of curve design based on cyclicity.
We assume that all the curves are cyclic.
Then on the iteration step we perform some \textit{local mutation} gathering~$2^d$ cycles into one cycle.
The local mutation here means the following:
\begin{itemize}
    \item in each cell we take some subcell and the next one in the cycle,
    \item for these cells we say that the corresponding next subcells are next to the chosen in the next cell,
    \item so, for the next subcell w.\,r.\,t. the chosen one is the subcell chosen in the previous cell.
\end{itemize}
In terms of graphs, we chain~$2^d$ edges into a cycle by~$2^d$ edges and then remove the initial edges.
If we take pairs of sub-cells in such a way that new edges connected by cycle are touching,
then we obtain a continuous surjective map going to the limit:
\[
    S^1 \to [0,1]^d.
\]
See examples of local mutations in~\S\ref{sec:hcurve}.

Usually one needs to perform some transformations during construction process to obtain continuity.
Sometimes it is not necessary as in case of Z-curve.
It is the simplest curve used, but it is not having continuity,
so it has bad locality preserving properties and  is not widely used.
Usually, any continuous curve gives better results,
but it that case the construction procedure needs carefully chosen reflections and/or rotations.

In the proposed construction of cyclic curves on the step of local mutation
in terms of traversal we may need to change the traversal direction and initial point.
Note that the cycle lengths are always degrees of $2$.
So, the change of initial point is simply the addition of cell index to a number
of new initial point modulo a degree of $2$ (addition of $d$-bit numbers),
and change of direction with change of initial point to the previous one (before reversal)
is simply a bitwise complement of d-bit number.

Of course, one may need to use some symmetries depending on the particular
curve construction algorithm.
In the section~\ref{sec:hcurve} we will see that a cyclic continuous curve can be constructed
for any dimension without usage of any symmetries.
An addition with maybe one bitwise complement is computationally cheaper than the evaluation
and application of symmetry.

    \section{H-curve}
    \label{sec:hcurve}
    \subsection{Construction}
\label{subsec:construction}

This section is devoted to the construction of cyclic fractal space-filling curve
for any $d>1$ without using symmetries.
For any dimension $d$, we will traverse half-sized cells in the initial cube in the same way.

Taking a $d$-bit number k as an index in traversal (counting from $0$),
we obtain the corresponding cell coordinate bits as consecutive bits of the number
\[
    g_d(k) := k\oplus \lfloor (k \mod 2^d)/2 \rfloor \mod 2^d
\]
(the symbol~$\oplus$ means bitwise sum, or, xor).
This function permutes the set~$\{0,\ldots,2^d-1\}$.
Therefore, the function~$g_d^{-1}$ is well-defined on the set~$\{0,\ldots,2^d-1\}$.

As we claimed, in the cells we do not apply any reflections or rotations to the cells and sub-cells.
For the curve construction we need only the local mutations.
For explicit computation of correspondence between indexes and cells we need to calculate
the index shifts and find all direction reversals.

\subsection{Local mutation}
\label{subsec:loc-mutation}

For convenience let us assume that grid cells are unit cubes, and the initial big cube has side length~$2^n$.

Actually, for any dimension~$d>1$ we will apply the same local mutation.
This mutation will always act on the central~$4 \times 2 \times \ldots \times 2$-parallelepiped.

\begin{lemma}
    Given~$d>1$, for any~$n\geqslant 2$ the restriction of the graph composed of $2^d$
    half-size cycles in the cube with side length~$2^n$ onto the central~$4\times 2 \times \ldots \times 2$-parallelepiped
    form the same graph, namely, if we denote its vertices
    with $\{0,1,2,3\}\times\{0,1\}^{d-1}$, then the edge set would be
    \[
        (\{0\}\times p, \{1\}\times p)\text{ and } (\{2\}\times p, \{3\}\times p)\text{ for all } p\in\{0,1\}^{d-1}.
    \]
\end{lemma}

\begin{proof}
    Assume that we have the grid of integral points in the cube~$[0,2^{n+1}-1]^d$,
    and we initially have the cyclic traversals of cubes of side~$2$.
    They form a grid of~$2^n$ cells.
    Then we consequently apply mutations gathering cycles into cycles
    traversing cells of sizes~$4, 8, \ldots, 2^{n+1}$.
    Each time we consider the central $4\times 2\times \ldots \times 2$-parallelepiped
    in some cell of size~$4,8,\ldots,2^{n+1}$, then
    each of these parallelepipeds has even minimal first coordinate and odd minimal other coordinates.
    This implies that the restrictions of initial $2^{nd}$ cycles on them
    are same and coincide with the written above graph.
    At the same time, these parallelepipeds have pairwise non-intersecting sets of vertices.
    Therefore, mutations of previous steps of construction do not affect the final step.
\end{proof}

On Fig.~\ref{fig:figure-d2n3step1},~\ref{fig:figure-d2n3step2},~\ref{fig:figure-d2n3step3}
we see examples of mutations.
On these figures we color some black edges red.
Then we draw a number of green edges such that together the green and red edges form cycles.
After the mutation we remove red edges and color green edges black.
Note that if in the red-green cycle we contract all the red edges, then we obtain
exactly the graph corresponding to the traversal of the cube of size~$2$ and
the same dimension.
Obviously, we will see the same behavior in any dimension.

\begin{example}
    \FloatBarrier
    Consider the case of~$d=2$ and~$n=3$ (side length~$8$).
    \begin{figure}[ht]
        \centering
        \includegraphics{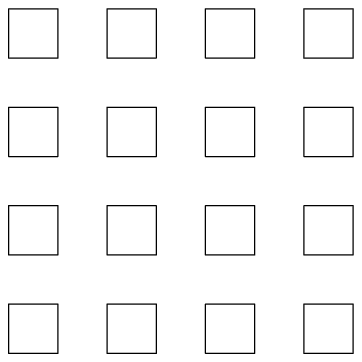}
        \qquad
        \includegraphics{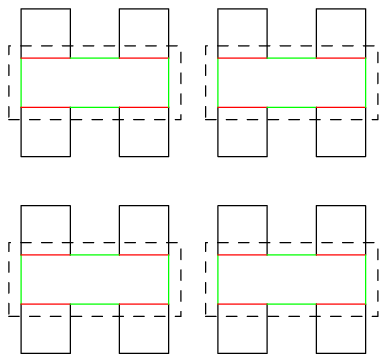}

        \caption{Join of cycles in squares of side~$2$ into cycles in squares of side~$4$.}
        \label{fig:figure-d2n3step1}
    \end{figure}
    \begin{figure}[ht]
        \centering
        \includegraphics{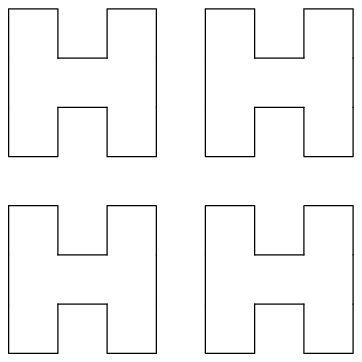}
        \qquad
        \includegraphics{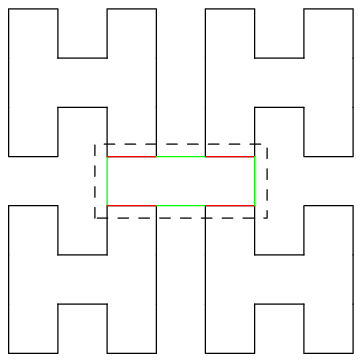}

        \caption{Join of cycles in squares of side~$4$ into cycles in squares of side~$8$.}
        \label{fig:figure-d2n3step2}
    \end{figure}
    \begin{figure}[ht]
        \centering
        \includegraphics{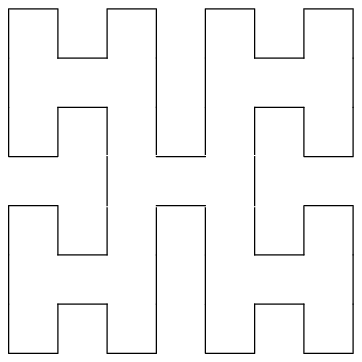}
        \caption{H-curve for $d=2$, $n=3$.}
        \label{fig:figure-d2n3step3}
    \end{figure}

    On fig.\,\ref{fig:figure-d2n3step1} we join cycles of side length~$2$,
    after on fig.\,\ref{fig:figure-d2n3step2} we join cycles of side~$4$,
    and on fig.\,~\ref{fig:figure-d2n3step3} we see the result.
    \FloatBarrier
\end{example}

\begin{definition}
    We call the constructed above family of curves \textit{H-curves} for all~$d>1, n$.
    Also, we will call \textit{H-curves} the limit curves for all~$d>1$.
\end{definition}

We name them this way for the form of the second iteration of plane curve.
Next iterations also looks like the letter `H', but more tangled and shaggy.
For $d \geqslant 3$ we can consider these curves as high-dimensional ``generalizations'' of letter `H'.

\begin{example}
    \FloatBarrier
    \begin{figure}
        \centering
        \includegraphics{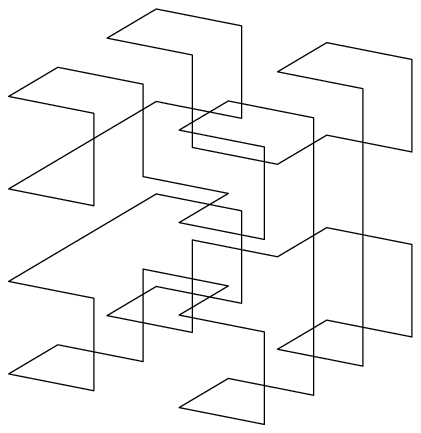}
        \qquad
        \includegraphics{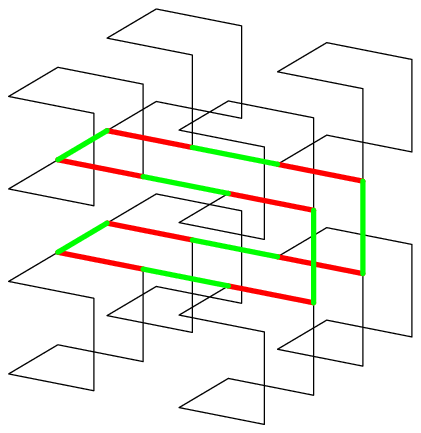}
        \caption{Example of H-curve for $d=3,n=2$ and the central mutation.}
        \label{fig:hc3,2}
    \end{figure}

    On fig.\,\ref{fig:hc3,2} we see the example of H-curve for~$d=3,n=2$ (side length~$4$) and how looks
    the mutation in three-dimensional case.
    In the higher dimensions it looks the same, but less illustrative.
    \FloatBarrier
\end{example}

\begin{theorem}
    For any~$n\in\mathbb{N}$ and any~$d>1$ the H-curve cyclically traverses all the unit cells.
    Each move to the next cell is a move to an adjacent cell.

    For any~$k<n$, for one cycle the curve one time enters and one time
    leaves any of cells of grid of~$2^{dk}$-side cells,
    and the traversal of these cells is~H-curve for the pair $(d,k)$.

    For $n \to \infty$ we can choose infinitely decreasing sequence of cells
    such that each one contains all the following,
    and we obtain the sequence that converges to the continuous map~$h\colon S^1 \to [0,1]^d$.
\end{theorem}

\begin{proof}
    The first part of statement is obvious by the construction of curve and by choice of the mutation.

    The second part is obvious for~$k=n-1$ by the iteration of construction and for any~$k$
    by induction from~$n-1$ down to~$1$.

    For the third part we need to take unit cells in such a way that
    \begin{itemize}
        \item for increasing~$n$ the matching between smaller cells and intervals is a subdivision of matching between
        larger cells and intervals;
        \item end points of~$[0,1]$ are mapped to the same point.
    \end{itemize}
    Actually, it is enough to take the first cell of initial subdivision of each next time to take the first sub-cell,
    where the cycle enters the cell.
    The condition that $0$ and $1$ are mapped into the same point is obvious.
    Continuity is standard and follows from the same reasons as for Hilbert and other curves.
\end{proof}

We will compare below properties of H-curve, Hilbert curve and Z-curve.

\begin{remark}
    Actually, for dimension~$2$ there is only one construction method of the Hilbert.
    As it was noticed in~\cite{16Curves}, for higher dimensions there are many ways to generalize the construction of
    the curve to any dimension such that its restriction to~$d=2$ gives the usual plane Hilbert curve.
    In~\cite{16Curves} there are $5$ ways to do this.
    The commonly used version seems to be called Butz-Hilbert curve in~\cite{16Curves}.
    For higher dimensions is seems that different variations of Hilbert curve would give
    very close results.
    At the same time, their constructions have the same complexity (computational and mathematical).
    So, we will compare H-curve with the commonly used Butz-Hilbert curve.
\end{remark}

\begin{remark}
    One of curves constructed in~\cite{16Curves} is called there \textit{inside-out curve}.
    For $n=2$, it returns to the cell adjacent to the initial point, but for bigger~$n$ it loses continuity.
    In this sense H-curve can be called an \textit{inside-out-repeat curve} as a curve moving from the center to the
    perimeter in one octant, back to the center, out into another octant and so on cyclically.
\end{remark}

On fig.\,\ref{fig:d2n4} we see Z-curve, Hilbert curve and H-curve on plane.
On fig.\,\ref{fig:d3n2} we see Z-curve, Butz--Hilbert curve and H-curve in the same axes.

\begin{figure}
    \centering
    \includegraphics[scale=0.8]{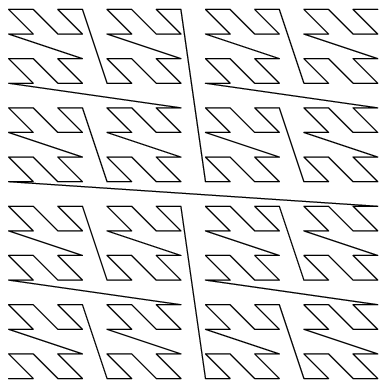}
    \quad
    \includegraphics[scale=0.8]{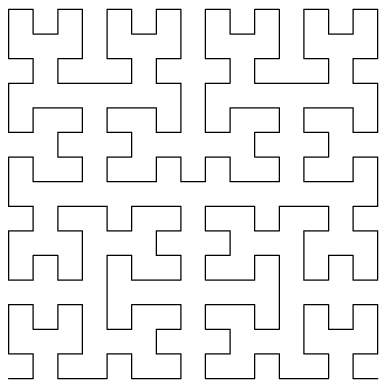}
    \quad
    \includegraphics[scale=0.8]{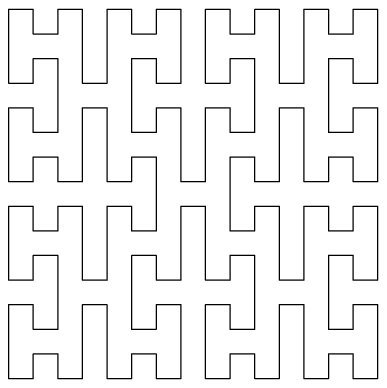}
    \caption{Z-curve, Hilbert curve, H-curve for $d=2,n=4$}
    \label{fig:d2n4}
\end{figure}

\begin{figure}
    \centering
    \includegraphics[scale=0.8]{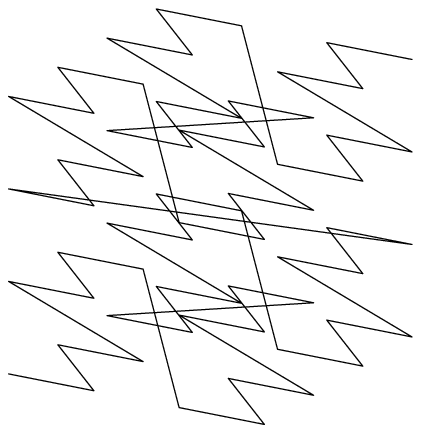}
    \includegraphics[scale=0.8]{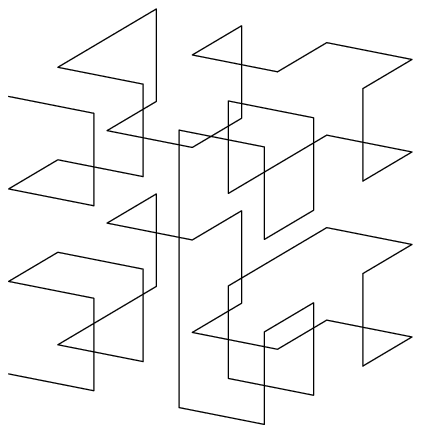}
    \includegraphics[scale=0.8]{pict-0.eps}
    \caption{Z-curve, Butz--Hilbert curve, H-curve for $d=3,n=2$}
    \label{fig:d3n2}
\end{figure}

\subsection{Index shifts and direction reversals}
\label{subsec:shifts-and-reversals}

Our next goal is to describe the correspondence between a unit cell with coordinates~$\overline{a}=(a_0,\ldots,a_d)$
in $d$-dimensional cube with side length~$2^n$ and its index~$r$ in the traversal along H-curve.
Say, we \textit{encode} the point~$\overline{a}$ by the index~$r$ and~\textit{decode} the index~$r$
to the point~$\overline{a}$.
So, we want to describe two mutually inverse functions
\[
    \xymatrix{
    \{0,\ldots,2^n-1\}^d \ar@/_/[rr]
    &&\mathbb{Z}/2^{nd}\mathbb{Z}. \ar@/_/[ll]
    }
\]
We will construct the functions recursively.
To make the construction easier, let us introduce some notation.

Let us write bits of $d$-bit numbers~$a_i$ into the matrix~$n\times d$ as rows.
Denote $n$-bit numbers in the rows of transposed matrix by~$\alpha^0,\ldots,\alpha^d$.
These coordinates are also known as coordinates in \textit{Z-order}.
It is easy to pass from $(a_i)$ to $(\alpha^j)$ and back, but~$\alpha^j$ are more convenient for algorithm design.
So, we will describe functions
\[
    \xymatrix{
    \{0,\ldots,2^d-1\}^n \ar@/_/[rr]_-{\mathtt{encode}}
    &&\mathbb{Z}/2^{nd}\mathbb{Z}. \ar@/_/[ll]_-{\mathtt{decode}}
    }
\]

Geometrical sense of~$\alpha$-coordinates corresponds to the iterative construction of curve.
Each cell can be coded by~$d$ bits of coordinates.
These~$d$ bits form the number~$\alpha^j$ for $j$-th iteration.
When subdividing the cube into the grid of~$2^d$ cells, we choose one of them
which has coordinates~$(\alpha^0,\ldots,\alpha^{d-1})$.

\begin{lemma}
    The central $4\times 2 \times \ldots \times 2$-parallelepiped in $d$-dimensional cube with side~$2^n$ consists of
    the set of points
    \[
        c_\alpha  = (\alpha, \overline{\alpha},\ldots, \overline{\alpha}) \text{ and }
        c_\alpha' = (\alpha, \overline{\alpha},\ldots, \overline{\alpha}\oplus 1)
    \]
    for all~$\alpha\in\mathbb{F}_2^n$.
\end{lemma}

\begin{proof}
    Denote the central cube of size~$2$ by~$C$ and the central~$4\times 2 \times \ldots \times 2$-parallelepiped by~$P$.
    Each half-size cube~$\alpha$ has a unique unit cell in~$C$.
    Denote it by~$c_\alpha$.
    Each half-size cube~$\alpha$ has two unit cells in~$P$,~$c_\alpha$ is on of them.
    Denote the other one by~$c_\alpha'$.
    The index~$\alpha$ for both~$c_\alpha$ and~$c_\alpha'$ corresponds to the
    first coordinate of f unit cell in~$\alpha$-coordinates.
    Our goal is to find remaining~$\alpha$-coordinates of these points.

    Bits of~$\alpha$ geometrically mean the choice of half-size cube in the first
    subdivision operation.
    To get the cell~$c_\alpha \in C$, we should take opposite coordinate choices
    for each coordinate on each next iteration.
    This exactly implies that all the next $\alpha$-coordinates equal~$\overline{\alpha}$.

    To take~$c_\alpha'$, we should take the same sub-cells until the last subdivision.
    At the last iteration we should change the first coordinate to get the adjacent
    cell along the first coordinate.
    This exactly means that~$c_\alpha'$ has all next coorinates equal~$\overline{\alpha}$
    until the last one which equals~$\overline{\alpha}\oplus 1$.
\end{proof}

\begin{corollary}
    The traversal of H-curve enters the half-size sub-cell~$\alpha$ at one of unit cells~$c_\alpha$ and~$c_\alpha'$ and
    leaves at other one.
\end{corollary}

We have fixed the traversal of the central~$4\times 2 \times \ldots \times 2$-parallelepiped.
Therefore, if we know the order of traversal of the pair~$(c_\alpha,c_\alpha')$, then we know if we need to reverse
the traversal of half-size sub-cell~$\alpha$.
(As it was noted above, they are neighbors in the sub-cell traversal).

Now we want to determine when the direction of half-size sub-cube traversal is either the same or opposite
to the direction of traversal of these sub-cubes.
Suppose the direction is the same.
Then before the mutation we pass from one cell of~$\{c_\alpha,c_\alpha'\}$ to another one,
so, traversing the remaining part of the cycle traversing the sub-cube, we pass them in the opposite
order, because the first one becomes the leaving unit cell,
and other one becomes the entering unit cell for sub-cube.
Vice versa, if the direction changes, then the order remains the same.
To avoid confusion, we consider the chain part traversing the sub-cube, but not whole the cube, because in the cycle
the proposition that a cell follows other one is nonsense.

\begin{lemma}
    Consider a $d$-dimensional cube of side length~$2^n$.
    In the construction of H-curve the edges $(0,\ldots,0)-(1,0,\ldots,0)$ and~$(2^n-1,\ldots,2^n-1)-(2^n-2,2^d-1,\ldots,2^d-1)$
    (denote them correspondingly $\overline{0}-\overline{0}'$ and $\overline{1}-\overline{1}'$)
    are passed in the same direction (along the first coordinate) for~$d$ odd and in the opposite direction for~$d$ even.
\end{lemma}

\begin{proof}
    The proof consists of two steps: to pass to~$n=1$ and to directly calculate for~$n=1$.

    At first, we pass to~$n=1$.
    Indeed, we obtain the traversal of the cube of size~$2^n$ by joining together traversals of~$2^{n-1}$ cubes
    with a central mutation.
    Note that the mutations does not affect the edge from/to corner vertices.
    So, for~$n>1$ the proposition is the same as for~$n=1$, and we can put~$n=1$ without loss of generality.

    Fix some~$d>1$.
    In Z-order the representations of the vertices are the following (we write the square brackets and index~$2$ to distinguish decimal and binary numbers):
    \begin{align*}
        \overline{0} = [\underbrace{0\ldots0}_d]_2, \quad & \overline{0}' = [\underbrace{0\ldots0}_{d-1}1]_2, \\
        \overline{1} = [\underbrace{1\ldots1}_d]_2, \quad & \overline{1}' = [\underbrace{1\ldots1}_{d-1}0]_2. \\
    \end{align*}
    Note that~$g(0)=[\underbrace{0\ldots0}_{d}]_2$ and~$g(1)=[\underbrace{0\ldots0}_{d-1}1]_2$.
    It only remains to find~$g^{-1}(\overline{1})$ and~$g^{-1}(\overline{1}')$.
    One of the following two cases holds:
    \begin{itemize}
        \item Let $d$ be even.
        Then
        \[
            g^{-1}([\underbrace{1\ldots1}_{d}]_2) = 2\cdot\frac{2^d-1}{3}, \quad
            g^{-1}([\underbrace{1\ldots1}_{d-1}0]_2) = 2\cdot\frac{2^d-1}{3} + 1.
        \]
        We see that~$\overline{1}'$ follows $\overline{1}$.
        \item Let $d$ be odd.
        Then
        \[
            g^{-1}([\underbrace{1\ldots1}_{d}]_2) = \frac{2^{d+1}-1}{3}, \quad
            g^{-1}([\underbrace{1\ldots1}_{d-1}0]_2) = \frac{2^{d+1}-1}{3} - 1.
        \]
        We see that~$\overline{1}$ follows $\overline{1}'$.
    \end{itemize}
    The calculations can be easily checked directly.
    This concludes the proof.
\end{proof}

\begin{corollary}
    For $d$ even there are no traverse reversals.
    For $d$ odd the only traverse reversal happens for~$n=2$.
\end{corollary}

\begin{proof}
    As we have seen above, the direction of a bigger cubes traversal
    from~$\overline{0}$ to~$\overline{0}'$
    is the same as for unit cells in cubes with side length~$2$ if in the cube of side length~$2$ the traversal
    of the edge $\overline{1}-\overline{1}'$ is opposite to the traversal direction
    of the edge $\overline{0}-\overline{0}'$.
    So, the direction for~$(n,d)$ for~$n>1$ is the same as for~$(1,d)$ for even~$d$ and opposite for odd~$d$.
    Therefore, there are no any reversals for~$d$ even and the only reversal for~$d$ odd is when~$n=2$.
    (For~$d$ odd and~$n>1$ the directions are opposite to the direction for~$n=1$, thus, they coincide.)
\end{proof}

\begin{theorem}
    For any $d>1$ and~$n\geqslant 1$ H-curve starts the traversal of sub-cell~$\alpha$
    at its unit sub-cell~$(\overline{\alpha}, \ldots, \overline{\alpha}, \overline{\alpha}\oplus p(\overline{\alpha}))$,
    and the direction of traversal changes if and only if~$d$ is odd and~$n=2$.
\end{theorem}

\begin{proof}
    Actually, it only remains to find which one of~$c_\alpha$ and~$c_\alpha'$ is the initial point.
    Note that the function~$g$ is $\mathbb{F}_2$-linear as a function~$g\colon \mathbb{F}_2^d \to \mathbb{F}_2^d$.
    Geometrically the operation~$\oplus\alpha$ corresponds to the composition of reflections along coordinate
    hyperplanes corresponding to bits equal~$1$ in~$\alpha$.
    Therefore, we can find only the initial point of the sub-cube corresponding to~$\alpha=[0\ldots0]_2$.
    In the traversal of this sub-cube (before the mutation) $c_\alpha$ and~$c_\alpha'$ follows each other.
    So, after the mutation the second one becomes the entering unit cell of a sub-cube, and first one becomes the leaving unit cell.
    From the reasoning above it follows that for~$d$ even the entering point is~$\overline{1}$ and for~$d$ odd
    the entering point is~$\overline{1}'$.
    Restoring generality of~$\alpha$ and due~$\mathbb{F}_2$-linearity, we can rewrite the initial point of sub-cube~$\alpha$
    with the parity function~$p$ as the point~$(\overline{\alpha}, \ldots, \overline{\alpha}, \overline{\alpha}\oplus p(\overline{\alpha}))$ in Z-order.
\end{proof}

\subsection{Algorithmic construction}
\label{subsec:alg}

Here we briefly describe algorithms of two functions:
\begin{itemize}
    \item[\texttt{encode}] which maps $d$-dimensional array of cell coordinates in the cube $\{0,\ldots,2^n-1\}^d$ to the index,
    \item[\texttt{decode}] performing the inverse function.
\end{itemize}

Here index means the number of cell in the traversal.
It can be considered as an arbitrary integer number or a number in $\{0,\ldots,2^{nd}-1\}$ due to $2^{nd}$-periodicity.

For convenience, we will evaluate coordinates in \textit{Z-order}: instead of $d$ $n$-bit numbers we consider
$n$ $d$-bits numbers composed of corresponding bits of coordinates.
If we write down $d$ $n$-bit numbers as rows of bit matrix, then the corresponding $n$ $d$-bit numbers in Z-order
become rows of the transposed matrix.

Denote the coordinates of the cell with index~$r$ by~$(a_0,\ldots,a_{d-1})$.
Denote the corresponding Z-order numbers by~$(\alpha^0,\ldots,\alpha^{n-1})$.

Denote~$g_n(k) = (k \mod 2^n)\oplus (\lfloor k/2 \rfloor \mod 2^{n-1})$.
Note that~$g_n$ is a bijection on the set~$\{0,\ldots,2^n-1\}$, so~$g_n^{-1}$ is well-defined on this set.
Denote by~$p(x)$ the \textit{parity} of~$x$, i.\,e. $1$ if the number of odd bits in~$x$ is odd and~$0$ otherwise.

\subsubsection{Encode}
Given dimension~$d$, depth~$n$, numbers $\overline{\alpha}=(\alpha^0,\ldots,\alpha^{n-1})$,
we calculate the index~$r = \mathtt{encode}(d, n, \overline{\alpha})$ as follows.
\begin{itemize}
    \item Put~$r_0 = g_d^{-1}(\alpha_0)$.
    \item Put~$r = \mathtt{encode}(d,n-1,(\alpha^1,\ldots,\alpha^{n-1}))$.
    \item Put~$r' = \mathtt{encode}(d, n-1, (e,\ldots,e\oplus p(e)))$, where~$e = (-1-r_0)\mod 2^d$ (bitwise complement).
    \item Return~$r_0\cdot 2^{d(n-1)} + (r - r' \mod 2^{d(n-1)})$.
\end{itemize}

\subsubsection{Decode}
Given dimension~$d$, depth~$n$, and index~$r$ we
calculate~$\overline{\alpha}=(\alpha^0,\ldots,\alpha^{n-1})$
with the function~$\mathtt{decode}(d, n, r, i=0)$ ($i$ is the argument with the default value~$0$)
as follows.
\begin{itemize}
    \item If~$i \geqslant n$, the function returns~$\overline{\alpha}=(\alpha^0,\ldots,\alpha^{n-1})$.
    \item Put $\alpha_i = g(\rho_i)$, where~$\rho_i = \lfloor r / 2^{d(n-1-i)} \rfloor$.
    \item Put~$r' = \mathtt{encode}(d, n-1, (e,\ldots,e\oplus p(e)))$, where~$e = (-1-r_0)\mod 2^d$ (bitwise complement).
    \item Put $r'' = r - r' \mod 2^{d(n-1)}$.
    \item $\mathtt{decode}(d, n-1, r'', i + 1)$.
\end{itemize}

\subsubsection{Tail recursion}

Here we see that each of \texttt{decode} and \texttt{encode} call two of these functions for smaller~$d$.
But one of these calls is a call to get the index of a corner of an~$(n,d)$-cube or an adjacent cell by the first coordinate.
In practice, we should keep more points than the number of corners of~$(n',d)$-cubes for~$n' < n$.
So they can be precomputed and stored (or lazily evaluated on demand), so the first call will require only~$O(1)$ operations asymptotically.
This improvement makes~\texttt{decode} and~\texttt{encode} tail recursive.

Of course, we can choose initial point other way (for example, put into correspondence the zero index
to the point with zero coordinates), but then we should apply the same additional corrections for mutations.
In the chosen way we always remove the edges with the same indexes.
So, actually, there is no significant difference.

With precomputed corner indexes and implementation of tail recursions as loops on~\texttt{C},
the profiling results of \texttt{encode} and~\texttt{decode} functions for pair~$(n,d)=(7,7)$
for a million calls are the following (see~Table~\ref{tab:profiling}).

\FloatBarrier

\begin{table}
    \begin{tabular}{|l|r|} \hline
    function & average time spent with function descendents, ms/call \\ \hline \hline
    \texttt{encode\_h} & $0.08$ \\
    \texttt{decode\_h} & $0.06$ \\
    \texttt{encode\_Hilbert} & $0.31$ \\
    \texttt{decode\_Hilbert} & $0.47$ \\ \hline
    \end{tabular}
    \caption{Profiling results}
    \label{tab:profiling}
\end{table}

\FloatBarrier

So, we can see that H-curve computes significantly faster than the Hilbert curve.

    \section{Clustering property}
    \label{sec:clustering_property}
        \subsection{Model definition}
\label{subsec:model}

We define and test clustering property following~\cite{908985}.

Let us described the model of experiment.

We assume that data space~$\mathcal{U}$ has dimension~$d$ and finite granularity,
say, a coordinate is an integer $n$-bit number.
So, $U = \{0,1,\ldots,2^n-1\}^d$.
Each point of the space corresponds to a grid cell.
A space-filling curve (below SFC for shortness) introduces a bijection~$\omega\colon U \to \{0,1,\ldots,2^{nd}-1\}$.
A \textit{query} is any subset $q\subseteq \mathcal{U}$.
Consider rectangular queries being intersections of coordinate half-spaces.
More generally (see~\cite{908985}), one can consider queries corresponding to connected simply connected domains.

\begin{remark}
    Here we understand~$\mathcal{U}$ as a subset of the lattice~$\mathcal{Z} = \mathbb{Z}^d$.
    We need some other identification of queries with geometrical objects
    to define connected and simply connected sets correctly.
    Namely, we consider the Euclidean space~$E = \mathcal{U}\otimes_\mathbb{Z}\mathbb{R}$.
    Consider a closed unit cube~$C$ in~$E$.
    It is a fundamental domain of the action~$\mathcal{Z} \lefttorightarrow E$.
    Given a query~$q$, denote by~$C_q$ the set
    \[
        C_q := \bigcup_{p\in q}(p + C) \subset E
    \]
    that consists of shifts of the cube~$C$ by all points of the query.
    We say that a query~$q$ is connected (or simply connected) if so is the interior of~$C_q$.

    For instance, a two point query $q=\{x,y\}$ is connected if and only if $C_q^\circ$ is connected, i.\,e.~$x$~and~$y$
    differ by~$1$ in one coordinate and coincide in all the others.
\end{remark}

\begin{definition}
    A subset~$p\subseteq q$ of a query is called a \textit{cluster} with respect to a SFC~$\omega$
    if it is a maximal subset such that the points (or cells) of~$p$ are numbered consequently by~$\omega$.
    We denote the number of clusters in~$q$ by~$c_q(\omega)$.
\end{definition}

\begin{definition}
    A \textit{clustering property} of a SFC~$\omega$ with respect to a (maybe parametric) class of queries~$\mathcal{Q}$
    as the average number~$c_{\mathcal{Q}}(\omega)$ of clusters
    in~$q\in \mathcal{Q}$ (or the limits/asymptotics of cluster number
    as a function in the parameters if exist).
\end{definition}

Of course, there are also implicit parameters being the space granularity parameter~$n$
and the distribution over~$\mathcal{Q}$.
Usually, for fixed parameters the set~$\mathcal{Q}$ is finite, and the distribution is assumed to be uniform.
If we specify a probabilistic measure~$\mu$ on~$\mathcal{Q}$, then
\[
    c_\mathcal{Q}(\omega) := \int_\mathcal{Q}c_q(\omega)d\mu.
\]

We consider the class of cubic queries~$\mathcal{Q}_\ell$ where~$\ell$ is the side length of cubes.
In~\cite{908985} there were considered parametric classes of queries of same shape parametrized by their scales.
Also, limit asymptotics of average cluster number of a shape (cubes, spheres and some others)
as a function in the scale were considered.

%However, usually we do not have to evaluate them completely to compare indexes due to the recursive construction.

\subsection{Simulation results}
\label{subsec:simulation}

Our main goal is to minimize number of disk accesses.
This number depends on capacity of disk pages, model of memory access,
some particular algorithms of access, insertion and deletion.
We omit the technical details and compute average number of clusters,
or \textit{continuous runs} over a subspace representing a query region.

In~\cite{908985} the analytical results for different curves were tested on different query shapes
and an increasing range of sizes.
Note that the number of different query shapes is exponential in the dimensionality.
Consequently, for a large grid space and high dimensionality, each simulation run may require
an excessively large number of queries.
So we restrict simulations for~$d=2,3,4$.

For a given query shape and size, we do not test all the query positions but perform
a statistical simulation by random sampling of queries.
For query shapes, we choose squares and cubes.
In~\cite{908985} the asymptotic and simulation results we shown to be very close and were considered
as identical from round-off errors.
Also, results coincided for different shapes in simulations and analytic calculation with asymptotics.
So, we consider only quadratic and cubic queries due to reliability of the estimation method.

The results of the experiment are listed in Table~\ref{tab:results}.
For $d=2$ we compare average number of clusters for $10000$ random queries on $1024 \times 1024$ grid
(in~\cite{908985} for~$d=2$ the grid is the same and there were~$200$ queries
for a given combination of shape and size).

\begin{table}
    \begin{tabular}{|c|rrr|}
        \hline
        \multicolumn{4}{|c|}{$d=2$} \\ \hline
        $\ell$ & Z & Hilbert & H \\
        \hline
        2  &   2.62 &  2.00 &  1.99 \\ % 2.6242 & 2.0008  & 1.9917 \\
        3  &   4.51 &  3.00 &  3.01 \\ % 4.5074 & 2.9989  & 3.0108 \\
        4  &   6.36 &  4.01 &  3.99 \\ % 6.3622 & 4.0113  & 3.9917 \\
        5  &   8.25 &  4.99 &  5.00 \\ % 8.2505 & 4.9941  & 5.0028 \\
        6  &  10.23 &  6.00 &  6.00 \\ % 10.2291 & 6.0002  & 5.9952 \\
        7  &  12.26 &  7.00 &  7.00 \\ % 12.2607 & 7.0016  & 7.003  \\
        8  &  14.23 &  8.03 &  8.00 \\ % 14.2333 & 8.0284  & 7.9984 \\
        9  &  16.14 &  9.01 &  9.02 \\ % 16.136  & 9.0091  & 9.0174 \\
        10 &  18.00 &  9.94 &  9.97 \\ % 17.9915 & 9.9439  & 9.9693 \\
        11 &  20.04 & 10.98 & 10.98 \\ % 20.0447 & 10.9831 & 10.9762 \\
        12 &  22.24 & 12.07 & 12.00 \\ % 22.2361 & 12.0689 & 12.002 \\
        13 &  24.06 & 12.99 & 12.99 \\ % 24.0567 & 12.9876 & 12.9904 \\
        14 &  26.04 & 14.00 & 14.00 \\ % 26.0403 & 13.9986 & 13.9916 \\
        15 &  28.17 & 15.04 & 15.02 \\ % 28.1696 & 15.0399 & 15.0236 \\
%        16 &  30.28 & 16.11 & 16.16 \\ % 30.2814 & 16.1057 & 16.1629 \\
        \hline
    \end{tabular}
    \begin{tabular}{|c|rrr|}
        \hline
        \multicolumn{4}{|c|}{$d=3$} \\ \hline
        $\ell$ & Z & Hilbert & H \\
        \hline
        2  &   5.34 &   4.02 &   4.00 \\ %  5.3492 &  4.0226 &  4 \\
        3  &  13.51 &   9.04 &   9.01 \\ % 13.5082 &  9.0446 &  9.0114 \\
        4  &  25.58 &  16.08 &  16.04 \\ % 25.5849 & 16.0829 & 16.0442 \\
        5  &  41.63 &  25.07 &  24.99 \\ % 41.6257 & 25.0687 & 24.9875 \\
        6  &  61.62 &  36.10 &  36.03 \\ % 61.6187 & 36.1006 & 36.0315 \\
        7  &  85.74 &  49.08 &  49.00 \\ % 85.7443 & 49.0767 & 49.0008 \\
        8  & 113.96 &  64.38 &  64.13 \\ % 113.955 & 64.3827 & 64.1291 \\
        9  & 145.76 &  80.90 &  81.00 \\ % 145.763 & 80.9006 & 80.9983 \\
        10 & 181.04 &  99.85 &  99.75 \\ % 181.036 & 99.8457 & 99.7518 \\
        11 & 221.63 & 120.50 & 120.85 \\ % 221.628 & 120.491 & 120.853 \\
        12 & 267.50 & 144.72 & 144.77 \\ % 267.501 & 144.72  & 144.768 \\
        13 & 314.00 & 169.28 & 169.21 \\ % 314.001 & 169.283 & 169.207 \\
        14 & 363.72 & 195.11 & 194.73 \\ % 363.72  & 195.113 & 194.732 \\
        15 & 421.75 & 225.17 & 224.99 \\ % 421.754 & 225.167 & 224.992 \\
%        16 & 482.29 & 256.53 & 256.33 \\ % 482.286 & 256.529 & 256.331 \\
        \hline
    \end{tabular}
    \begin{tabular}{|c|rrr|}
        \hline
        \multicolumn{4}{|c|}{$d=4$} \\ \hline
        $\ell$ & Z & Hilbert & H \\
        \hline
        2  &   10.74 &    7.95 &    8.05 \\ % 10.736  &  7.9505 &  8.0494 \\
        3  &   40.49 &   26.96 &   26.98 \\ % 40.4942 & 26.9609 & 26.9754 \\
        4  &  102.33 &   64.39 &   64.14 \\ % 102.328 & 64.389  & 64.1382 \\
        5  &  208.39 &  125.23 &  125.01 \\ % 208.389 & 125.227 & 125.007 \\
        6  &  372.55 &  216.60 &  217.18 \\ % 372.547 & 216.604 & 217.18  \\
        7  &  600.43 &  343.52 &  343.02 \\ % 600.431 & 343.524 & 343.021 \\
        8  &  911.06 &  513.73 &  512.52 \\ % 911.062 & 513.725 & 512.519 \\
        9  & 1312.09 &  730.78 &  729.02 \\ % 1312.09  & 730.784 & 729.018 \\
        10 & 1810.43 &  991.12 &  995.21 \\ % 1810.43  & 991.115 & 995.212 \\
        11 & 2440.48 & 1331.96 & 1331.06 \\ % 2440.48  & 1331.96 & 1331.06 \\
        12 & 3185.88 & 1734.03 & 1728.66 \\ % 3185.88 & 1734.03 & 1728.66
        13 & 4080.00 & 2203.18 & 2197.00  \\ % 4080    & 2203.18 & 2197
        14 & 5091.67 & 2732.45 & 2726.83 \\ % 5091.67 & 2732.45 & 2726.83
        15 & 6329.08 & 3378.49 & 3375.01 \\ \hline
    \end{tabular}
    \caption{Average number of clusters in cubic queries with the cube side~$\ell$
        for $d=2,3,4$.}
    \label{tab:results}
\end{table}

\FloatBarrier

    \section{Conclusion}
    \label{sec:conclusion}
    In this paper we introduced a new way to construct cyclic space-filling curves.
    A particular simple family of curves is created (we call them H-curves).
    This family has a very close clustering property to Hilbert curves.
    At the same time, their construction is simpler and significantly faster.
    So, for a number of applications H-curves may be preferable than Hilbert curves.

    \appendix
    \section{Implementation}
    \label{sec:implementation}
        Let us introduce some notation used in pseudocode below:
    \begin{itemize}
        \item $d$ denotes the dimension,
        \item $n$ denotes the depth,
        \item $r$ denotes the number of cube in the traversal,
        \item $\ll_{cycle}$ and $\gg_{cycle}$ denote left and right cyclic bit shifts,
    \end{itemize}

    Implementation of Hilbert curve from~\cite{HPsCd} (rewritten):

    \begin{algorithmic}
        \Function{decode}{$n, d, r$}
            \For{$i\gets [0..d-1]$}
                \State{$J_i \gets \log_2(2\rho_i+1)-1$}
                \State{$\sigma^i \gets \rho^i \oplus \rho^i/2$}
                \If{$\rho^i \% 2$}
                    \State{$\tau_i \gets (\rho^i-1)\oplus(\rho^i-1)/2$}
                \Else
                    \State{$\tau_i \gets (\rho^i-2)\oplus(\rho^i-2)/2$}
                \EndIf
%                \State{$\tau^i \gets (\rho^i - \rho^i\%2 - 1)\oplus(\rho^i - \rho^i\%2 - 1)/2$}
                %\State{$\tau^i \gets \text{ if } \rho^i\%2 \text{ then } (\rho^i-1)\oplus(\rho^i-1)/2 \text{ else } (\rho^i-2)\oplus(\rho^i-2)/2 $}
                \State{$\widetilde{\sigma}^i \gets \sigma^i \gg_{cycle} J_0+\ldots+J_{i-1}$}
                \State{$\widetilde{\tau}^i \gets \tau^i \gg_{cycle} J_0+\ldots+J_{i-1}$}
                \State{$\omega^i \gets \text{ if } i=0 \text{ then } 0 \text{ else } \omega^{i-1}\oplus\widetilde{\tau}^{i-1}$}
                \State{$\alpha^i \gets \omega^i \oplus \widetilde{\sigma}^{i-1}$}
            \EndFor
        \EndFunction
    \end{algorithmic}

    \begin{algorithmic}
        \Function{encode}{$n, d, \alpha$}
            \For{$i \gets [0..d-1]$}
                \State{$\omega^i \gets \text{ if } i=0 \text{ then } 0 \text{ else } \omega^{i-1}\oplus\widetilde{\tau}^{i-1}$}
                \State{$\widetilde{\sigma}^i \gets \text{ if } i=0 \text{ then } \alpha_0 \text{ else } \alpha_i\oplus \omega^{i-1}$}
                \State{$\sigma^i \gets \widetilde{\sigma}^i \ll_{cycle} J_1+\ldots+J_{i-1}$}
                \For{$j \gets [0..n]$}
                    \State{$\rho_j^i \gets \text{ if } j=0 \text{ then } \sigma_0^i \text{ else } \sigma_j^i\oplus\sigma_{j-1}^i$}
                \EndFor
                \State{$J_i \gets \log_2(2\rho^i+1)-1$}
                \State{$\tau^i \gets \text{ if odd } parity(\sigma^i) \text{ then } \sigma^i\oplus 2^{n-1} \text{ else } \sigma^i\oplus 2^{n-1}\oplus 2^{J_i}$}
                \State{$\widetilde{\tau}_i \gets \tau_i \ll_{cycle} J_0+\ldots+J_{i-1}$}
            \EndFor
        \EndFunction
    \end{algorithmic}

    Implementation of H-curve:
    \begin{algorithmic}
        \Function{decode}{$n, d, r$}
            \For{$i \gets  [0..n-1]$}
                \State{$\alpha_i \gets g(\lfloor r / 2^{d(n-1)} \rfloor)$}
                \State{
                $r \gets r +
                \mathtt{encode}(
                n-1, d,
                (\underbrace{\overline{\alpha_i},\ldots,\overline{\alpha_i}}_{n-2},\overline{\alpha_i}\oplus p(\overline{\alpha_i}))
                \mod 2^{d(n-1)}
                $}
                \If{$d$ odd and~$n=2$}
                    \State{$r \gets -1-r \mod 2^{2d}$}
                \EndIf
            \EndFor
        \EndFunction
    \end{algorithmic}

    \begin{algorithmic}
        \Function{encode}{$n, d, \overline{\alpha}$}
            \State{$r \gets g^{-1}(\alpha_0)$}
            \State{
            $
            r' \gets \mathtt{encode}(n-1, d, (\alpha_1,\ldots,\alpha_{n-1}))$}\State{$r' \gets r' -
            \mathtt{encode}(
            n-1, d,
            (\underbrace{\overline{\alpha_i},\ldots,\overline{\alpha_i}}_{n-2},\overline{\alpha_i}\oplus p(\overline{\alpha_i}))
            \mod 2^{m(n-1)}
            $
            }
            \If{$d$ odd and~$n=2$}
                \State{$r' \gets -1-r' \mod 2^{2d}$}
            \EndIf\par
            \Return{$r\cdot 2^{d(n-1)} + r'$}
        \EndFunction
    \end{algorithmic}

    \bibliographystyle{unsrt}
    \bibliography{main-eng}

\end{document}